\documentclass[pdflatex,sn-basic]{sn-jnl}

\usepackage{nicefrac}
\usepackage[ruled]{algorithm2e}
\usepackage{float}
\usepackage{longtable}
\usepackage{array}
\usepackage{wrapfig}
\usepackage{colortbl}
\usepackage{threeparttablex}
\usepackage[normalem]{ulem}
\usepackage{makecell}
\SetKwComment{Comment}{}{}
\SetFuncSty{MyFuncSty}
\SetCommentSty{MyCommentSty}

\usepackage{graphicx}%
\usepackage{multirow}%
\usepackage{amsmath,amssymb,amsfonts}%
\usepackage{amsthm}%
\usepackage{mathrsfs}%
\usepackage[title]{appendix}%
\usepackage{xcolor}%
\usepackage{textcomp}%
\usepackage{manyfoot}%
\usepackage{booktabs}
\usepackage{algpseudocode}%
\usepackage{listings}%

\newtheorem{theorem}{Theorem}
\newtheorem{definition}{Definition}%

\begin{document}

\title[Article Title]{Modeling Complex Life Systems: Bayesian Inference for Weibull Failure Times Using Adaptive MCMC}

\author*[1]{\fnm{Tobias} \sur{Oketch}}\email{too21@msstate.edu}

\author[2]{\fnm{Mohammad} \sur{Sepehrifar}}\email{msepehrifar@math.msstate.edu}

\affil*[1,2]{\orgdiv{Department of Mathematics and Statistics}, \orgname{Mississippi State University}, \orgaddress{\street{175 President Circle}, \city{Mississippi State}, \postcode{39762}, \state{MS}, \country{USA}}}

\abstract{
This research develops a Bayesian framework for analyzing failure times using the Weibull distribution, addressing challenges in prior selection due to the lack of conjugate priors and multi-dimensional sufficient statistics. We propose an adaptive semi-parametric MCMC algorithm for lifetime data analysis, employing a hierarchical Bayesian model and the No-U-Turn Sampler (NUTS) in STAN. Twenty-four combinations of prior distributions are evaluated, with a noninformative LogNormal hyper-prior ensuring flexibility. A simulation study of seventy-two datasets with varying structures compares MCMC and classical methods, identifying optimal priors for Bayesian regularization. The approach effectively handles the Increasing Hazard Rate (IHR) and Decreasing Hazard Rate (DHR) scenarios. Finally, we demonstrate the algorithm's utility by predicting the remaining lifetime of prostate cancer patients, showcasing its practical application. This work advances Bayesian methodologies for modeling complex life systems and testing processes.
}



\maketitle

\section{Introduction}\label{sec1}

Predicting failure times in dynamic systems is crucial for reliability engineering and medical prognosis, yet existing methods often struggle with small data and evolving failure patterns. The Weibull distribution, renowned for its flexibility, is a key tool in addressing these challenges, as it effectively models diverse failure behaviors observed in engineering, biomedical research, and industrial reliability (see \cite{monteiro2017weibull}, \cite{wais2017two}, \cite{milad2011climate}, \cite{pieters2009comparison}).

Let \( t_1, t_2, \dots, t_n \) denote independently and identically distributed (i.i.d) failure times of \( n \) units. The random variable \( T \), representing the failure time, follows a Weibull distribution \( Weibull\left(\theta \right) \), where \( \theta = \left(\beta, \alpha\right) \) denote the shape parameter \( \beta \) and the scale parameter \( \alpha \). Its cumulative distribution function (CDF) is given by:

\begin{equation}
\label{WeibullCDF}
  F_T\left(t\right) = 
  \begin{cases}
  1 - \exp\left\{- {\left(\nicefrac{t}{\alpha}\right)}^{\beta}\right\} , & t, ~ \alpha, ~ \beta ~ > ~ 0 \\
  0                                                                & \text{otherwise}
  \end{cases}
\end{equation}

The shape parameter \( \beta \) is pivotal in characterizing system behavior: \( \beta = 1 \) indicates constant failure rates, \( \beta > 1 \) suggests increasing hazard rates (IHR) indicative of an aging system, and \( \beta < 1 \) corresponds to decreasing hazard rates (DHR) reflecting early-life failures.
This interpretation aligns with the Weibull hazard function:
\begin{equation}
\label{failureRateFunction}
h_T\left(t\right)=\nicefrac{F^{\prime}\left(t\right)}{1-F\left(t\right)}=\left(\nicefrac{\beta}{\alpha}\right)\left(\nicefrac{t}{\alpha}\right)^{\beta-1}
\end{equation}

This function describes how failure likelihood evolves and is crucial for predictive maintenance and risk assessment. The significance of \( \beta \) in interpreting system reliability has been extensively studied. For instance, \cite{chumnaul2018generalized} and \cite{chumnaul2022modified} developed robust methodologies for inference, highlighting the importance of accurate parameter estimation. Despite the widespread use of the Weibull distribution, parameter estimation remains challenging.

Classical methods such as the Method of Moments, Maximum Likelihood Estimation (MLE), and Ordinary Least Squares (OLS) often suffer from unstable estimates, especially when failure data is sparse or heavily skewed \cite{goldstein2004problems}. Bayesian inference offers a more robust alternative, yet its application to Weibull analysis is hindered by the lack of conjugate priors for Weibull parameters and the multi-dimensional nature of the shape parameter's sufficient statistic that necessitates advanced computational techniques. Markov Chain Monte Carlo (MCMC) methods provide a natural solution for these challenges, but they require a Bayesian framework with well-specified prior distributions \cite{gelman2017prior, depaoli2020importance, tian2023specifying}. Selecting appropriate priors remains a significant challenge, particularly for small or complex datasets, as traditional Bayesian approaches often lack conjugate priors, leading to inefficiencies and biases in inference.

\subsection{Bayesian Analysis}

We introduce a robust adaptive Bayesian framework for Weibull failure-time analysis to overcome these limitations. Given the observed failure times  \( t_1, t_2, \dots, t_n \), we assume:
\begin{equation}
t_i \sim Weibull\left(\beta, \alpha \right)
\end{equation}
Since the likelihood function for the Weibull distribution is:
\begin{align}\label{WeibLikelihoodFun}
	L\left\{\beta, \alpha ~ | ~ t_1, \dots t_n \right\} = \prod_{i=1}^{n}t_i^{\beta - 1}\left\{\frac{\beta}{\alpha^{\beta}} \right\}^n \exp\left\{-\frac{1}{\alpha^\beta}\sum_{i=1}^{n}t_i^\beta,\right\} 
\end{align}
direct estimation of \(\beta\) and \(\alpha\) requires computational techniques that handle the nonlinearity and lack of conjugacy. To ensure efficient posterior sampling, we adopt a hierarchical Bayesian approach, assigning flexible priors to \(\beta\) and \(\alpha\), and employ the No-U-Turn Sampler (NUTS) \cite{hoffman2014no}, an advanced variant of Hamiltonian Monte Carlo (HMC) \cite{betancourt2016diagnosing}, implemented in STAN \cite{carpenter2017stan}. This framework enables efficient posterior sampling and improved parameter estimation, particularly in complex, data-limited scenarios.

\begin{theorem} Convergence of Adaptive MCMC for Weibull Parameters

Let \( \theta = (\beta, \alpha) \) be the parameters of a Weibull distribution and let \( \pi(\theta \mid \mathbf{t}) \) be the posterior distribution given observed failure times \( \mathbf{t} = (t_1, t_2, \dots, t_n) \). Under the following conditions:
\begin{enumerate}
\item The prior distributions \( \pi(\beta) \) and \( \pi(\alpha) \) are proper and continuous on \( \mathbb{R}^+ \).
\item The likelihood function \( f(\mathbf{t} \mid \theta) \) is bounded and Lipschitz continuous in \( \theta \).
\end{enumerate}
Then, the adaptive MCMC algorithm using the No-U-Turn Sampler (NUTS) converges to the true posterior distribution \( \pi(\theta \mid \mathbf{t}) \) as the number of iterations \( N \to \infty \).
\end{theorem}

\begin{proof} The convergence follows from the ergodicity of adaptive MCMC algorithms under mild regularity conditions (see \cite{roberts2007coupling}). As a Hamiltonian Monte Carlo (HMC) variant, the NUTS algorithm satisfies detailed balance and ensures efficient parameter space exploration, even in high-dimensional settings \cite{hoffman2014no}. The boundedness and continuity of the likelihood function guarantee a unique stationary distribution to which the Markov chain converges.
\end{proof}

\subsubsection{Methodology}

The methodology introduced in this study is designed to overcome the limitations of traditional failure-time analysis methods, particularly in scenarios lacking previous data on similar systems or experiencing unique aging conditions. Our approach involves the following key steps:

\begin{enumerate}
\item Prior Specification: We assign four priors to \( \beta \) (LogNormal, HalfNormal, Gamma, Exponential) and six priors to \( \alpha \) (LogNormal, Gamma, InverseGamma, HalfCauchy, HalfNormal, Exponential), resulting in 24 prior combinations.

\item Posterior Sampling: Using NUTS, we construct Markov chains to sample from the posterior distribution \( \pi(\theta \mid \mathbf{t}) \), where \( \theta = (\beta, \alpha) \). The posterior is defined hierarchically as:
\[
f(\theta \mid \mathbf{t}, \gamma) \propto f(\mathbf{t} \mid \theta) \pi(\theta \mid \gamma) h(\gamma),
\]
with \( \gamma \) following a LogNormal(0, 25) hyper-prior.

\item Simulation Study: We generate 72 survival datasets from Weibull distributions with varying structures and fit the 24 MCMC models to each dataset. This allows us to identify optimal prior combinations that regularize the Bayesian model effectively.

\item Model Comparison: We evaluate model performance using weighted relative efficiency, WRE, comparing Bayesian and classical methods (MLE, Moments, OLS).

\item Application: We apply the proposed algorithm to prostate cancer data, validating our simulation results and predicting remaining lifetimes.

\end{enumerate}

\begin{theorem}\label{theo2}
	Let \(\theta = (\beta, \alpha)\) be the shape and scale parameters of a Weibull distribution, estimated from a sample of size \(n\). Under regular conditions, the Maximum Likelihood Estimator \(\hat{\theta}_n = (\hat{\beta}_n, \hat{\alpha}_n)\) is asymptotically normal:
	\[
	\sqrt{n} (\hat{\theta}_n - \theta_0) \overset{d}{\longrightarrow} N(0, I_n(\theta_0)^{-1}),
	\]
	where \(I_n(\theta)\) is the Fisher Information Matrix, and \(\theta_0 = (\beta_0, \alpha_0)\) represents the true parameter values.
\end{theorem}

\begin{proof}
	See Appendix \ref{Ap-theo2}
\end{proof}

We compute the WRE in Equation \ref{WeightedRelEfficiency} using the asymptotic efficiency result from Theorem \ref{theo2}, as given by the variance-covariance matrix in Equation \ref{AsympVarCovMatrix}.

\begin{align}\label{WeightedRelEfficiency}
	WRE(\hat{\theta}) =\frac{ \nicefrac{\sigma^2_i(\hat{\theta})}{ \sum_{i=1}^{27} \sigma^2_i(\hat{\theta}) } }{ \nicefrac{ V_i(\hat{\theta}) }{\sum_{i=1}^{27}V_i(\hat{\theta})  } }.
\end{align}

where, \newline $V_i(\hat{\theta}) = \nicefrac{1}{n}\left( 1.1087\nicefrac{\hat{\alpha}^2}{\hat{\beta}^2}  +  0.6079 \hat{\beta}^2 - 2 \times   0.2570 \hat{\alpha}\right)$ is the total asymptotic variance and $\sigma^2_i(\hat{\theta}) = \sigma^2_i(\hat{\beta}) + \sigma^2_i(\hat{\alpha})$ the total sampling variance of the $i^{th}$ model fitted to a given data. The asymptotic variance-covariance matrix is given by:

\begin{align}\label{AsympVarCovMatrix}
Var(\hat{\theta}) = \left[I_n \left(\theta\right)\right]^{-1} = \frac{1}{n}
\begin{pmatrix}
1.1087\nicefrac{\alpha^2}{\beta^2} & 0.2570 \alpha \\ 
0.2570 \alpha & 0.6079 \beta^2
\end{pmatrix}
\end{align}

This research bridges theoretical advancements in Bayesian inference with practical applications, offering a scalable and adaptive framework for modeling failure-time data across diverse real-world settings.


\section{Simulation Study}\label{StudyProcedure}

To evaluate the effectiveness of the proposed Bayesian methodology for estimating the shape parameter ($\beta$) and the scale parameter ($\alpha$) of the Weibull distribution, we conduct a comprehensive simulation study. We examine \(24\) different prior combinations, considering four priors for $\beta$ (LogNormal, HalfNormal, Gamma, Exponential) and six priors for $\alpha$ (LogNormal, Gamma, InverseGamma, HalfCauchy, HalfNormal, Exponential), as outlined in Table~\ref{tab: DefinitionOfPriors} in Appendix \ref{MCMCMethods}. This diverse prior specification allows a rigorous assessment of the method's robustness and adaptability in lifetime data analysis.

We generate \(72\) datasets from a Weibull distribution, simulating two distinct hazard rate behaviors: IHR and DHR. The sample sizes ($n$) considered are \(15, 25, 55\), and \(100\), enabling an evaluation of the Bayesian method's scalability in small and large sample settings. Each dataset is analyzed using \(24\) Bayesian models corresponding to the prior combinations and compared against three classical methods: MLE, Method of Moments, and OLS.

We conduct a posterior sampling using the No-U-Turn Sampler (NUTS) within the STAN framework. This method is an efficient variant of Hamiltonian Monte Carlo (HMC) optimized for high-dimensional parameter spaces. We run each MCMC chain for a total of \(2,000\) iterations, which includes a burn-in period of \(1,000\) iterations. We use the R-hat criterion to assess convergence, ensuring that all chains meet the convergence threshold of less than \(1.01\).

As an illustrative case, consider the Gamma-Gamma prior combination for the shape ($\beta$) and scale ($\alpha$) parameters, respectively. We specify the Bayesian model as:

\begin{align}
p(\beta, \alpha, A_{\beta,\alpha}, B_{\beta,\alpha} |\mathbf{t}) & \propto \text{Weibull}(.|\beta, \alpha) \times \text{Gamma}(.|A_{\beta,\alpha}, B_{\beta,\alpha}) \times \text{LogNormal}(.|0, 25),
\end{align}

And the hyper-priors and priors as:

\begin{eqnarray}
\begin{aligned}
A_{\beta}, B_{\beta} & \sim \text{LogNormal}(0,25) & \text{(hyper-priors)} \\
A_{\alpha}, B_{\alpha} & \sim \text{LogNormal}(0,25) \\
\beta & \sim \text{Gamma}(A_{\beta}, B_{\beta}) & \text{(priors)} \\
\alpha & \sim \text{Gamma}(A_{\alpha}, B_{\alpha}) \\
\mathbf{t} & \sim \text{Weibull}(t_{i} | \beta, \alpha), \quad i = 1, 2, \dots, n. & \text{(Likelihood)}
\end{aligned}
\end{eqnarray}

The parameters $A_{\beta}$, $B_{\beta}$, $A_{\alpha}$, and $B_{\alpha}$ are derived using the moment-matching technique applied to the bootstrap distribution of $\boldsymbol{\theta}_B = \{ \hat{\theta}_1, \dots, \hat{\theta}_{1000} \}$. The complete set of prior specifications is provided in Tables \ref{tab: ShapeScalePriorSpecs} and \ref{tab: ScaleOnlyPriorSpecs} in Appendix \ref{MCMCMethods}. Algorithm \ref{alg: WeibullAlgorithmcovers} describes how to characterize the distribution of Weibull parameters. The MCMC technique uses the sampling distributions of the parameters to initialize the hyperpriors by the Moments matching technique.

The performance of the Bayesian models is assessed relative to classical methods using Equation \ref{WeightedRelEfficiency}, the WRE metric, as defined in Theorem \ref{theo2}. The WRE accounts for sampling variance and asymptotic efficiency, providing a comprehensive criterion for model comparison.

Finally, to validate the simulation results, we apply the Bayesian framework to a prostate cancer dataset, demonstrating the method's practical utility in survival analysis. This real-world application highlights the robustness of the Bayesian approach, particularly in cases where data may be sparse or contain substantial noise. Figure \ref{fig: ComputationFramework} outlines the general computational framework used in this research.

\bigskip
 
\begin{minipage}{\linewidth} 
\begin{algorithm}[H]
\DontPrintSemicolon
    \caption{MCMC - Classic Methods for the Distribution of Weibull Parameters}
\label{alg: WeibullAlgorithmcovers}
    \SetKwFunction{BootSample}{SampleWithReplacement}
    \SetKwFunction{EstimateWeibull}{EstimateWeibullParameters}
    \SetKwFunction{SummarizePosterior}{ComputePosteriorSummary}
    \SetKwFunction{SummarizeClassic}{ComputeClassicSummary}
    \SetKwFunction{PostSampler}{PosteriorSampler}
    \SetKwFunction{MCMC}{MCMC}
    \SetKwFunction{ModelData}{CreateModelData}
    \SetKwFunction{InitHPar}{InitializeHyperParameters}
    \SetKwData{Classic}{ClassicMethod}
    \KwIn{
    $\mathcal{T} = \{T_{\left( i\right)} : i= 1, \dots n\}, \quad \textit{Ordered life data}$ \\
    \hspace{.05cm} $B, I, P \in \mathbb{N} \quad \textit{Numbers of bootstraps, MCMC iterations, \& parameters}$
    }
    \KwOut{$\mathrm{E_{\hat{\theta}}, \sigma^2_{\hat{\theta}}, V_{\hat{\theta}} }, ~ \mathrm{E_{\Tilde{\theta}}, \sigma^2_{\Tilde{\theta}}, V_{\Tilde{\theta}}}\quad \textit{Bootstrap (Classic) \& posterior (MCMC) parameter estimates}$
    }
    $\textit{Let }\mathrm{A_1}[B, 2] \quad \textit{Array of $B\times 2$ doubles}$ \newline
    \For{$b\gets 1 ~ \KwTo ~ B$} {
    $\mathrm{S_b} \gets \BootSample{$\mathcal{T}$}$\;
    $\mathrm{\hat{\theta}_{b}}\gets \EstimateWeibull{$S_b$}\quad \textit{MLE, Moments, OLS methods}$
    $\mathrm{A_1\left[b , ~ \right] \gets \hat{\theta}_b}\quad \textit{Each bootstrap produces P parameter results}$
    }
    $\mathrm{ClassicValues} \gets \SummarizeClassic{$A_1$} \quad \mathrm{\left[\mathrm{E_{\hat{\theta}}, \sigma^2_{\hat{\theta}}, V_{\hat{\theta}} }\right]}$ \newline
    \If{\Classic}
    {
     \Return $\mathrm{ClassicValues}$\;
    }
    $\textit{Let } \mathrm{A_2[\mathrm{I, ~2}]} \quad \textit{ Array of $I \times 2$ doubles}$ \newline
    $\MCMC{\dots} \quad \textit{Specify MCMC sampling model}$
    $\gamma_0 \gets \InitHPar{$E_{\hat{\theta}}, \sigma_{\hat{\theta}}$}$\; 
    $\mathrm{d} \gets \ModelData{$\mathcal{T}$}$\;
    \For{$i\gets 1 ~ \KwTo ~ I$} {
    $\mathrm{A_2\left[i, ~ \right]} \gets \PostSampler{$\sim\MCMC, \gamma_0$, $\mathrm{d}, \dots$}$
    }
    $\mathrm{PosteriorValues} \gets \SummarizePosterior{$A_2$},~ \left[\mathrm{E_{\Tilde{\theta}}, \sigma^2_{\Tilde{\theta}}, V_{\Tilde{\theta}}}\right]$\newline
    \Return $\mathrm{PosteriorValues}$\;
\end{algorithm}
\end{minipage}

\begin{figure}[H]
	\centerline{\includegraphics[width = 5in]{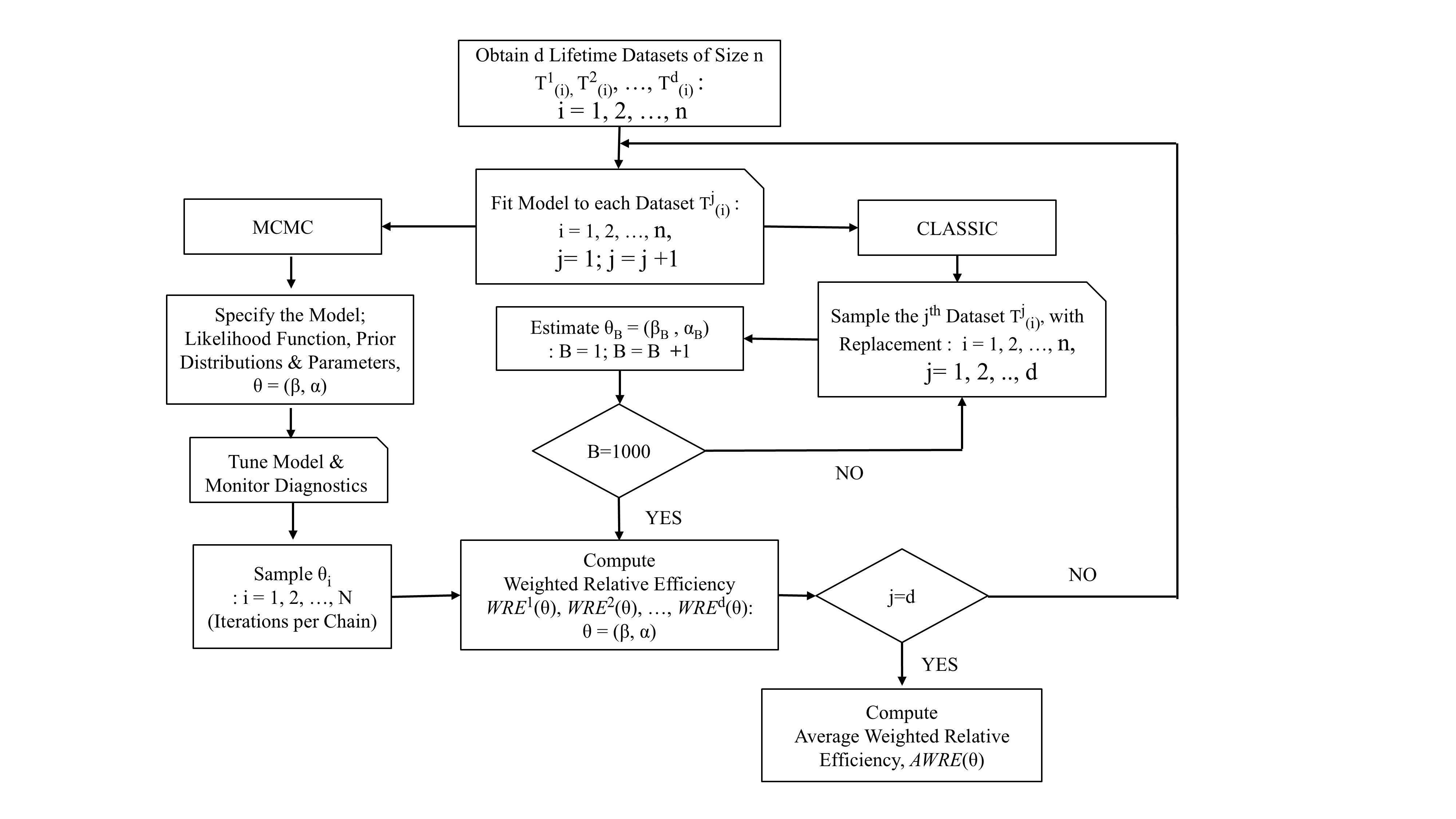} }
	\caption{\label{fig: ComputationFramework} The flow chart illustrates the computational framework used in this research to develop and compare the efficiency of the classic and MCMC models in estimating parameters of the $\mathit{Weibull}$ distribution. The average weighted relative efficiency (AWRE) is the average of the weighted relative efficiency (WRE) for lifetime datasets with either decreasing or increasing hazard rate properties.}
\end{figure}


\subsection{Analysis of Framework and the Algorithm: Model Selections and Priorities }

Figure \ref{fig: Graph_EfficiencySmallAndLargeSamples} illustrates efficiency trends when fitting the Weibull distribution model to small and large sample sizes of lifetime data, comparing the top MCMC and classical methods. Table \ref{tab: ContrastMethodEfficiencies} contrasts the behavior of classical and MCMC methods under small and large sample conditions:

\begin{figure}[!htb]	
		\centerline{\includegraphics[width=1.5\linewidth, height= .6\linewidth, keepaspectratio]{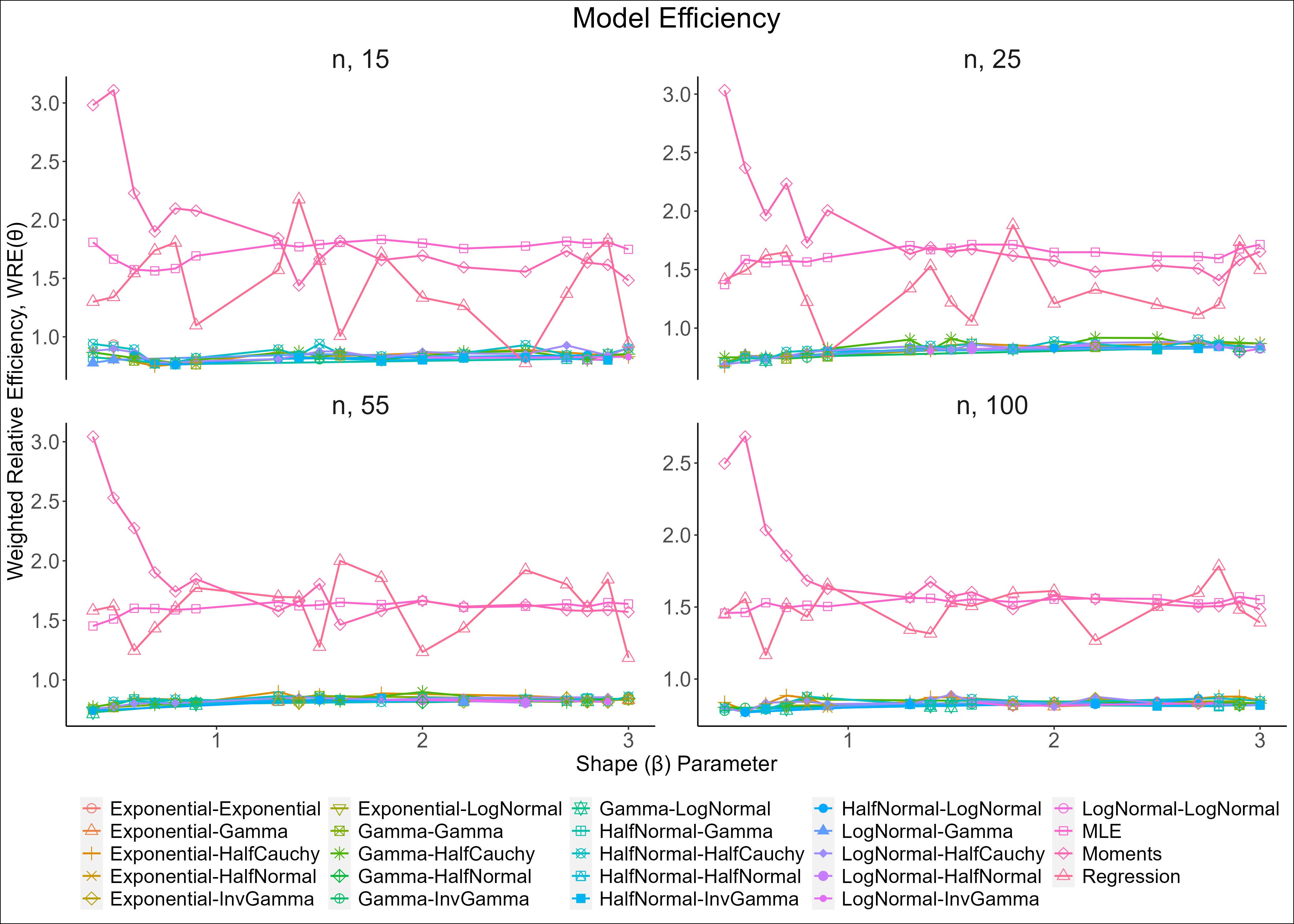}  }
		\caption{ \label{fig: Graph_EfficiencySmallAndLargeSamples}Efficiency trends for small and large samples of Weibull-distributed data with monotone hazard rates. MCMC methods consistently outperform classical methods.
		}	
\end{figure}
		
 \begin{table}[h!]\label{tab: ContrastMethodEfficiencies}
	\caption{Comparison of method efficiencies for small and large sample sizes}
	\fontsize{8}{10}\selectfont
	\centering
	\begin{tabular}{|l|l|l|}
	\hline
	\textbf{Method}       			& \textbf{Efficiency (Small Samples)} 			& \textbf{Efficiency (Large Samples)} \\ \hline
	\textbf{Moments Approach} 	& Improves as shape parameter ($\beta$) 		& Consistent efficiency for datasets\\ 
									& increases to 1, then remains  					& with increasing hazard rate  \\
							    	     	& constant for $\beta > 1$       				&   properties                     \\ \hline
	\textbf{MLE Method}         		& Constant efficiency for datasets 				& More consistent efficiency  \\ 
						         	    	& with both decreasing and      					& for datasets with increasing  \\
						          	    	& increasing hazard rates properties				& hazard rate properties			\\ \hline
	\textbf{Regression Method}    	& Unpredictable efficiency from  				& Less unpredictability compared \\ 
									& sample to sample								& to small samples				  \\ \hline
	\textbf{MCMC Methods} 	   	& Consistently efficient for datasets 				& Superior efficiency compared to \\ 
									&with both decreasing and						& classical methods					\\
									& increasing hazard rate properties				&										\\ \hline
	\end{tabular}
\end{table}

Classical methods such as the Moments approach improve in efficiency as the shape parameter (\(\beta\)) approaches 1, remaining constant for \(\beta >1\). The MLE method exhibits consistent efficiency, while the Regression method shows sample-to-sample variability. In contrast, MCMC methods demonstrate consistent efficiency across all datasets, offering better reliability.

In Bayesian inference, the choice of prior distributions significantly influences the posterior distribution and the resulting inferences. For Weibull-distributed survival data, the efficiency of the Bayesian model relates to the properties of the priors selected for the shape (\(\beta\)) and scale (\(\alpha\)) parameters. Due to its heavy-tailed nature, the HalfCauchy prior is particularly effective for wide-scale parameters. At the same time, the Exponential, LogNormal, and Gamma priors constrain the shape parameter, regularizing it based on expected hazard rate behavior.

For large sample sizes with increasing hazard rate properties, we recommend using Exponential, LogNormal, Gamma, or HalfNormal priors (in that order) for the shape parameter and a HalfCauchy prior for the scale parameter. For datasets with decreasing hazard rate properties, we recommend the Exponential, HalfNormal, Gamma, or LogNormal priors (in that order) for the shape parameter and HalfCauchy for the scale parameter. LogNormal or HalfNormal priors for the shape parameter and Gamma priors for the scale parameter also deliver desirable efficiency.

MCMC estimators with carefully chosen priors achieve asymptotic efficiency for large sample sizes,  meaning that the variance of the MCMC estimators converges to the Cramér-Rao lower bound, outperforming classical methods, outperforming classical methods, as shown in Tables \ref{Tab: WeightedRelativeEfficiency_SmallSamples_15Units} and \ref{Tab: WeightedRelativeEfficiency_SmallSamples_100Units}. 
The HalfNormal-HalfCauchy combination yields the highest efficiency across datasets, as shown in Tables \ref{Tab: WeightedRelativeEfficiency_SmallSamples_15Units}. 
Table \ref{Tab: AverageWeightedRelativeEfficiency_SmallSamples} illustrates that the Exponential-Exponential prior combination is most effective for small samples (around 15 units) with decreasing hazard rates, and HalfNormal-HalfCauchy priors are optimal for increasing hazard rates. It also highlights other effective MCMC prior combinations, all demonstrating superior performance compared to classical methods when applied to Weibull-distributed data. 

For large sample sizes, the Exponential-HalfCauchy combination achieves the best performance for the shape parameter around 1.4, and the LogNormal-HalfCauchy and HalfNormal-HalfCauchy combinations perform best for shape parameters around 1.5 and 1.6, respectively, as shown in Table \ref{Tab: WeightedRelativeEfficiency_SmallSamples_100Units}. Table \ref{Tab: AverageWeightedRelativeEfficiency_LargeSamples} compares the AWRE of top MCMC and classical methods for large samples, highlighting that MCMC methods are consistently more efficient.

Empirical studies on simulated datasets demonstrate that MCMC methods outperform classical approaches when applied to models addressing increasing and decreasing hazard rates. This superiority is attributed to MCMC's capacity to capture the underlying data structure better and reduce posterior variance through various prior combinations. For large sample sizes, the Exponential-HalfCauchy achieves the best performance in terms of the AWRE for both decreasing and increasing hazard rates. Conversely, under small sample sizes, the HalfNormal-HalfCauchy and Gamma-HalfCauchy perform best for the decreasing and increasing hazard rates, respectively.


\begin{table}[ht]
\centering
\caption{\label{Tab: WeightedRelativeEfficiency_SmallSamples_15Units} 
	Weighted relative efficiency (WRE) of the top \(5\) MCMC priors vs. classical methods for \(n=15 \), comparing various prior combinations for the Weibull shape and scale parameters.
}
\fontsize{8}{10}\selectfont
\begin{tabular}[ht]{rrlllr}
\toprule
\multicolumn{3}{c}{ } & \multicolumn{2}{c}{Prior Distribution} & \multicolumn{1}{c}{Model Efficiency} \\
\cmidrule(l{3pt}r{3pt}){4-5} \cmidrule(l{3pt}r{3pt}){6-6}
n & $\text{shape, }\beta$ & method & $\text{shape, }\beta$ & $\text{scale, }\alpha$ & $WRE(\theta)$\\
\midrule
 &  & Moments & - & - & 2.982\\

 &  & MLE & - & - & 1.808\\

 &  & Regression & - & - & 1.300\\

 &  &  & HalfNormal & HalfCauchy & 0.939\\

 &  &  & LogNormal & HalfCauchy & 0.880\\

 &  &  & Gamma & HalfCauchy & 0.867\\

 &  &  & HalfNormal & Gamma & 0.826\\

 & \multirow[t]{-6.50}{*}[3\dimexpr\aboverulesep+\belowrulesep+\cmidrulewidth]{\raggedleft\arraybackslash 0.4} & \multirow[t]{-5}{*}{\raggedright\arraybackslash MCMC} & LogNormal & Gamma & 0.777\\

 &  & Moments & - & - & 3.109\\

 &  & MLE & - & - & 1.664\\

 &  & Regression & - & - & 1.340\\

 &  &  & Exponential & Exponential & 0.934\\

 &  &  & HalfNormal & HalfCauchy & 0.917\\

 &  &  & LogNormal & HalfCauchy & 0.896\\

 &  &  & Exponential & HalfCauchy & 0.823\\

 & \multirow[t]{-6.50}{*}[3\dimexpr\aboverulesep+\belowrulesep+\cmidrulewidth]{\raggedleft\arraybackslash 0.5} & \multirow[t]{-5}{*}{\raggedright\arraybackslash MCMC} & LogNormal & Gamma & 0.803\\

 &  & Moments & - & - & 2.229\\

 &  & MLE & - & - & 1.575\\

 &  & Regression & - & - & 1.545\\

 &  &  & HalfNormal & HalfCauchy & 0.892\\

 &  &  & LogNormal & HalfCauchy & 0.867\\

 &  &  & HalfNormal & LogNormal & 0.840\\

 &  &  & Gamma & HalfCauchy & 0.819\\

\multirow[t]{-18.60}{*}[11\dimexpr\aboverulesep+\belowrulesep+\cmidrulewidth]{\raggedleft\arraybackslash 15} & \multirow[t]{-6.58}{*}[3\dimexpr\aboverulesep+\belowrulesep+\cmidrulewidth]{\raggedleft\arraybackslash 0.6} & \multirow[t]{-5}{*}{\raggedright\arraybackslash MCMC} & Gamma & Gamma & 0.797\\
\bottomrule
\end{tabular}
\end{table}


\begin{table}[ht]
\centering
\caption{\label{Tab: AverageWeightedRelativeEfficiency_SmallSamples}Average weighted relative efficiency (AWRE) of the top 5 MCMC priors vs. classical methods for \( n = 15 \) and \( n = 25 \), averaging WRE across lifetime datasets with varying hazard rates. Overall, MCMC methods are more efficient.
}
\centering
\fontsize{8}{10}\selectfont
\begin{tabular}[ht]{rllllr}
\toprule
\multicolumn{3}{c}{ } & \multicolumn{2}{c}{Prior Distribution} & \multicolumn{1}{c}{Model Efficiency} \\
\cmidrule(l{3pt}r{3pt}){4-5} \cmidrule(l{3pt}r{3pt}){6-6}
n & hazard & method &  $\text{shape, }\beta$ & $\text{scale, }\alpha$ & $AWRE(\theta)$\\
\midrule
 &  & Moments & - & - & 2.399\\

 &  & MLE & - & - & 1.648\\

 &  & Regression & - & - & 1.471\\

 &  &  & Exponential & Exponential & 0.934\\

 &  &  & HalfNormal & HalfCauchy & 0.853\\

 &  &  & LogNormal & HalfCauchy & 0.833\\

 &  &  & Gamma & HalfCauchy & 0.800\\

 & \multirow[t]{-6.5}{*}[3\dimexpr\aboverulesep+\belowrulesep+\cmidrulewidth]{\raggedright\arraybackslash DHR} & \multirow[t]{-5}{*}{\raggedright\arraybackslash MCMC} & HalfNormal & LogNormal & 0.782\\

 &  & MLE & - & - & 1.791\\

 &  & Moments & - & - & 1.645\\

 &  & Regression & - & - & 1.441\\

 &  &  & HalfNormal & HalfCauchy & 0.862\\

 &  &  & LogNormal & HalfCauchy & 0.856\\

\multirow[t]{-10.60}{*}[7\dimexpr\aboverulesep+\belowrulesep+\cmidrulewidth]{\raggedleft\arraybackslash 15} & \multirow[t]{-4.50}{*}[3\dimexpr\aboverulesep+\belowrulesep+\cmidrulewidth]{\raggedright\arraybackslash IHR} & \multirow[t]{-3}{*}{\raggedright\arraybackslash MCMC} & Gamma & HalfCauchy & 0.845\\
\cmidrule{1-6}
 &  & Moments & - & - & 2.224\\

 &  & MLE & - & - & 1.543\\

 &  & Regression & - & - & 1.365\\

 &  &  & Gamma & HalfCauchy & 0.762\\

 &  &  & LogNormal & HalfCauchy & 0.749\\

 &  &  & HalfNormal & HalfCauchy & 0.743\\

 &  &  & Gamma & Gamma & 0.738\\

 & \multirow[t]{-6.50}{*}[3\dimexpr\aboverulesep+\belowrulesep+\cmidrulewidth]{\raggedright\arraybackslash DHR} & \multirow[t]{-5}{*}{\raggedright\arraybackslash MCMC} & LogNormal & Gamma & 0.729\\

 &  & MLE & - & - & 1.666\\

 &  & Moments & - & - & 1.585\\

 &  & Regression & - & - & 1.360\\

 &  &  & Gamma & HalfCauchy & 0.872\\

 &  &  & HalfNormal & HalfCauchy & 0.837\\

\multirow[t]{-10.60}{*}[7\dimexpr\aboverulesep+\belowrulesep+\cmidrulewidth]{\raggedleft\arraybackslash 25} & \multirow[t]{-4.50}{*}[3\dimexpr\aboverulesep+\belowrulesep+\cmidrulewidth]{\raggedright\arraybackslash IHR} & \multirow[t]{-3}{*}{\raggedright\arraybackslash MCMC} & LogNormal & HalfCauchy & 0.833\\
\bottomrule
\end{tabular}
\end{table}


\begin{table}[ht]
	\scriptsize
\centering
\caption{\label{Tab: WeightedRelativeEfficiency_SmallSamples_100Units} Weighted relative efficiency of the top \( 5 \) MCMC priors vs. classical methods for \( n = 100 \), showing MCMC methods' consistent superior efficiency over classical methods. 
	}
\centering
\fontsize{8}{10}\selectfont
\begin{tabular}[t]{rrlllr}
\toprule
\multicolumn{3}{c}{ } & \multicolumn{2}{c}{Prior Distribution} & \multicolumn{1}{c}{Model Efficiency} \\
\cmidrule(l{3pt}r{3pt}){4-5} \cmidrule(l{3pt}r{3pt}){6-6}
n & shape & method & $\text{shape, }\beta$ & $\text{scale, }\alpha$ & $WRE(\theta)$\\
\midrule
 &  & Moments & - & - & 1.675\\

 &  & MLE & - & - & 1.561\\

 &  & Regression & - & - & 1.316\\

 &  &  & Exponential & HalfCauchy & 0.873\\

 &  &  & LogNormal & HalfCauchy & 0.855\\

 &  &  & HalfNormal & Gamma & 0.821\\

 &  &  & HalfNormal & HalfCauchy & 0.817\\

 & \multirow[t]{-6.50}{*}[3\dimexpr\aboverulesep+\belowrulesep+\cmidrulewidth]{\raggedleft\arraybackslash 1.4} & \multirow[t]{-5}{*}{\raggedright\arraybackslash MCMC} & Gamma & LogNormal & 0.814\\

 &  & Moments & - & - & 1.573\\

 &  & MLE & - & - & 1.538\\

 &  & Regression & - & - & 1.528\\

 &  &  & LogNormal & HalfCauchy & 0.893\\

 &  &  & Exponential & HalfCauchy & 0.874\\

 &  &  & Exponential & Gamma & 0.859\\

 &  &  & HalfNormal & HalfCauchy & 0.838\\

 & \multirow[t]{-6.50}{*}[3\dimexpr\aboverulesep+\belowrulesep+\cmidrulewidth]{\raggedleft\arraybackslash 1.5} & \multirow[t]{-5}{*}{\raggedright\arraybackslash MCMC} & Gamma & LogNormal & 0.812\\

 &  & Moments & - & - & 1.604\\

 &  & MLE & - & - & 1.554\\

 &  & Regression & - & - & 1.508\\

 &  &  & HalfNormal & HalfCauchy & 0.862\\

 &  &  & LogNormal & HalfCauchy & 0.846\\

 &  &  & LogNormal & LogNormal & 0.838\\

 &  &  & Gamma & inverseGamma & 0.828\\

\multirow[t]{-18.60}{*}[11\dimexpr\aboverulesep+\belowrulesep+\cmidrulewidth]{\raggedleft\arraybackslash 100} & \multirow[t]{-6.50}{*}[3\dimexpr\aboverulesep+\belowrulesep+\cmidrulewidth]{\raggedleft\arraybackslash 1.6} & \multirow[t]{-5}{*}{\raggedright\arraybackslash MCMC} & HalfNormal & Gamma & 0.823\\
\bottomrule
\end{tabular}
\end{table}


\begin{table}[htbp]
\centering
\caption{\label{Tab: AverageWeightedRelativeEfficiency_LargeSamples} Average weighted relative efficiency of the top 5 MCMC priors vs. classical methods for \( n = 55 \) and \( n = 100 \), averaging WRE across lifetime datasets with varying hazard rates.}
\fontsize{8}{10}\selectfont
\begin{tabular}[t]{rllllr}
\toprule
\multicolumn{3}{c}{ } & \multicolumn{2}{c}{Prior Distribution} & \multicolumn{1}{c}{Model Efficiency} \\
\cmidrule(l{3pt}r{3pt}){4-5} \cmidrule(l{3pt}r{3pt}){6-6}
n & hazard & method & $\text{shape, }\beta$ & $\text{scale, }\alpha$ & $AWRE(\theta)$\\
\midrule
 &  & Moments & - & - & 2.223\\

 &  & MLE & - & - & 1.559\\

 &  & Regression & - & - & 1.543\\

 &  &  & HalfNormal & HalfCauchy & 0.802\\

 &  &  & Exponential & HalfCauchy & 0.795\\

 &  &  & Gamma & HalfCauchy & 0.791\\

 &  &  & LogNormal & HalfCauchy & 0.783\\

 & \multirow[t]{-6.50}{*}[3\dimexpr\aboverulesep+\belowrulesep+\cmidrulewidth]{\raggedright\arraybackslash DHR} & \multirow[t]{-5}{*}{\raggedright\arraybackslash MCMC} & HalfNormal & Gamma & 0.782\\

 &  & MLE & - & - & 1.635\\

 &  & Regression & - & - & 1.630\\

 &  & Moments & - & - & 1.610\\

 &  &  & Exponential & HalfCauchy & 0.848\\

 &  &  & Gamma & HalfCauchy & 0.827\\

\multirow[t]{-10.60}{*}[7\dimexpr\aboverulesep+\belowrulesep+\cmidrulewidth]{\raggedleft\arraybackslash 55} & \multirow[t]{-4.50}{*}[3\dimexpr\aboverulesep+\belowrulesep+\cmidrulewidth]{\raggedright\arraybackslash IHR} & \multirow[t]{-3}{*}{\raggedright\arraybackslash MCMC} & HalfNormal & HalfCauchy & 0.824\\
\cmidrule{1-6}
 &  & Moments & - & - & 2.064\\

 &  & MLE & - & - & 1.494\\

 &  & Regression & - & - & 1.462\\

 &  &  & Exponential & HalfCauchy & 0.836\\

 &  &  & LogNormal & HalfCauchy & 0.819\\

 &  &  & Gamma & HalfCauchy & 0.816\\

 &  &  & HalfNormal & HalfCauchy & 0.797\\

 & \multirow[t]{-6.50}{*}[3\dimexpr\aboverulesep+\belowrulesep+\cmidrulewidth]{\raggedright\arraybackslash DHR} & \multirow[t]{-5}{*}{\raggedright\arraybackslash MCMC} & LogNormal & Gamma & 0.778\\

 &  & MLE & - & - & 1.550\\

 &  & Moments & - & - & 1.549\\

 &  & Regression & - & - & 1.494\\

 &  &  & Exponential & HalfCauchy & 0.845\\

 &  &  & LogNormal & HalfCauchy & 0.835\\

\multirow[t]{-10.60}{*}[7\dimexpr\aboverulesep+\belowrulesep+\cmidrulewidth]{\raggedleft\arraybackslash 100} & \multirow[t]{-4.50}{*}[3\dimexpr\aboverulesep+\belowrulesep+\cmidrulewidth]{\raggedright\arraybackslash IHR} & \multirow[t]{-3}{*}{\raggedright\arraybackslash MCMC} & HalfNormal & HalfCauchy & 0.831\\
\bottomrule
\end{tabular}
\end{table}

MCMC estimation is computationally slower than the MLE due to model compilation and iterative sampling, with runtime increasing linearly with dataset size. However, MCMC incorporates prior information, making it more robust in small samples and providing stable estimates and complete posterior distributions for uncertainty quantification.

Conversely, MLE is computationally efficient and asymptotically unbiased, but it can be unstable with small samples or complex likelihood surfaces, as shown in Tables \ref{Tab: WeightedRelativeEfficiency_SmallSamples_15Units}, \ref{Tab: AverageWeightedRelativeEfficiency_SmallSamples}, \ref{Tab: WeightedRelativeEfficiency_SmallSamples_100Units}, and \ref{Tab: AverageWeightedRelativeEfficiency_LargeSamples}. Despite its slower speed, MCMC offers more robustness and flexibility for hierarchical models and prior-based inference.

We propose a graphical posterior predictive check to assess model fit, particularly in small sample sizes, by comparing observed data to simulated data from the posterior predictive distribution. A well-fitting model should produce simulated data that closely resembles the observed data, indicating an adequate representation of the underlying data-generating process.

\newpage
\clearpage

\section{Application: Prostate Cancer Survival Data}\label{ProstCancerResults}

We illustrate the proposed Algorithm \ref{alg: WeibullAlgorithmcovers} using the prostate cancer dataset from Table \ref{tab: SurvivalTimesData}, initially presented by \cite{hollander1979testing}. The dataset comprises survival times for 90 prostate cancer patients, recorded from diagnosis $(t = 0)$. Three individuals with zero survival time are excluded, leaving 87 observations for analysis. The \textit{Epstein test} confirms an increasing hazard rate, justifying the use of the Weibull model.
\begin{table}[ht]
\caption{\label{tab: SurvivalTimesData} Survival times (months) of 90 prostate cancer patients.}
\centering
\begin{tabular}{ccccccccccccccc}
\toprule
0 & 0 & 0 & 2 & 3 & 4 & 6 & 7 & 7 & 8 & 9 & 9 & 11 & 11 & 11 \\
12 & 12 & 12 & 15 & 15 & 16 & 16 & 16 & 17 & 17 & 18 & 19 & 19 & 20 & 21\\
22 & 22 & 23 & 24 & 25 & 25 & 26 & 26 & 26 & 27 & 27 & 28 & 28 & 29 & 29\\
30 & 31 & 32 & 32 & 32 & 33 & 33 & 34 & 35 & 36 & 37 & 37 & 38 & 40 & 41\\
41 & 42 & 42 & 43 & 45 & 45 & 45 & 46 & 47 & 47 & 48 & 48 & 51 & 53 & 53\\
54 & 54 & 57 & 60 & 61 & 62 & 62 & 67 & 69 & 87 & 97 & 97 & 100 & 145 & 158\\
\bottomrule
\end{tabular}
\end{table}

Table \ref{tab: OutcomesProstateCancerDt} presents the top ten prior combinations for the Weibull shape and scale parameters, comparing them against classical estimation methods. MCMC consistently outperforms traditional approaches, with the LogNormal-Gamma prior yielding the most accurate results.


\begin{table}[ht]
\centering
\caption{\label{tab: OutcomesProstateCancerDt} Comparison of the efficiency of top \(10\) MCMC priors vs. classical methods for estimating Weibull parameters in prostate cancer data. 
The actual parameter values are $\beta= 1.43$ and $\alpha =40.34$, 
pooled from the results of the 27 fitted models.}
\centering
\fontsize{8}{10}\selectfont
\begin{tabular}[t]{@{}ccllccc@{}}
\toprule
\multicolumn{2}{c}{ } & \multicolumn{2}{c}{Prior Distribution} & \multicolumn{2}{c}{Estimated Parameters} & \multicolumn{1}{c}{Model Efficiency} \\
\cmidrule(l{3pt}r{3pt}){3-4} \cmidrule(l{3pt}r{3pt}){5-6} \cmidrule(l{3pt}r{3pt}){7-7}
n & method & shape, $\beta$ & scale,  $\alpha$ & shape,  $\hat{\beta}$ & scale,  $\hat{\alpha}$ & $WRE(\hat{\theta})$\\
\midrule
 & Regression & - & - & 1.583 & 40.095 & 1.820\\

 & MLE & - & - & 1.457 & 40.391 & 1.566\\

 & Moments & - & - & 1.367 & 39.742 & 1.330\\

 &  & LogNormal & Gamma & 1.433 & 40.360 & 0.854\\

 &  & Gamma & HalfCauchy & 1.436 & 40.489 & 0.845\\

 &  & Exponential & HalfNormal & 1.417 & 40.402 & 0.832\\

 &  & Exponential & HalfCauchy & 1.416 & 40.368 & 0.831\\

 &  & HalfNormal & HalfCauchy & 1.440 & 40.480 & 0.829\\

 &  & Gamma & HalfNormal & 1.434 & 40.404 & 0.816\\

 &  & HalfNormal & LogNormal & 1.436 & 40.366 & 0.808\\

 &  & Exponential & inverseGamma & 1.418 & 40.240 & 0.805\\

 &  & Exponential & LogNormal & 1.413 & 40.232 & 0.798\\

\multirow[t]{-12}{*}[3\dimexpr\aboverulesep+\belowrulesep+\cmidrulewidth]{\raggedleft\arraybackslash 87} & \multirow[t]{-10}{*}{\raggedright\arraybackslash MCMC} & HalfNormal & HalfNormal & 1.439 & 40.459 & 0.795\\
\bottomrule
\end{tabular}
\end{table}

Figure \ref{fig: PlotErrBarProstateData} illustrates the uncertainty in shape parameter estimates, showing that MCMC methods provide more precise estimates than classical methods. As a result, using the MCMC techniques proposed here allows for more confident computation of key survival analysis metrics, such as the Mean Residual Lifetime (MRL) of prostate cancer patients. Figure \ref{fig: MRLPlotProstateData} presents the MRL estimates. MCMC-based methods yield less biased estimates than MLE and regression, which underestimate and overestimate MRL, respectively.

\begin{figure}[ht]
\centerline{\includegraphics[width=1.5\linewidth, height= .6\linewidth, keepaspectratio]{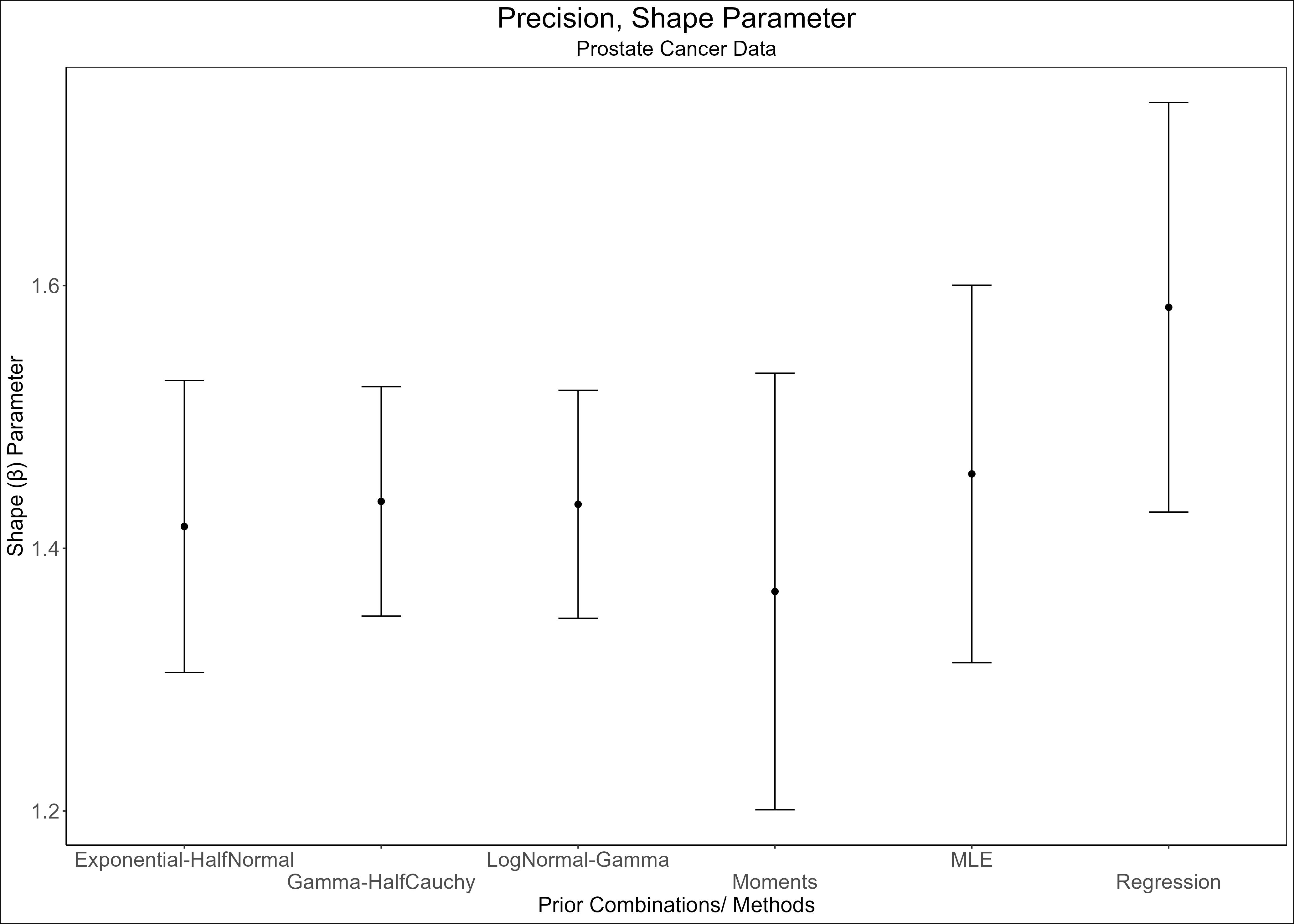}  }
\caption{Comparison of the uncertainty of the top 3 MCMC-based models vs. classical models for the shape parameter of Weibull-distributed prostate cancer survival times. We compute the interval as two empirical deviances from the actual parameter value obtained using an integrated approach where we pool the results from all methods under the study. \label{fig: PlotErrBarProstateData}} 
\end{figure}

\begin{figure}[ht]
\centerline{\includegraphics[width=1.5\linewidth, height= .6\linewidth, keepaspectratio]{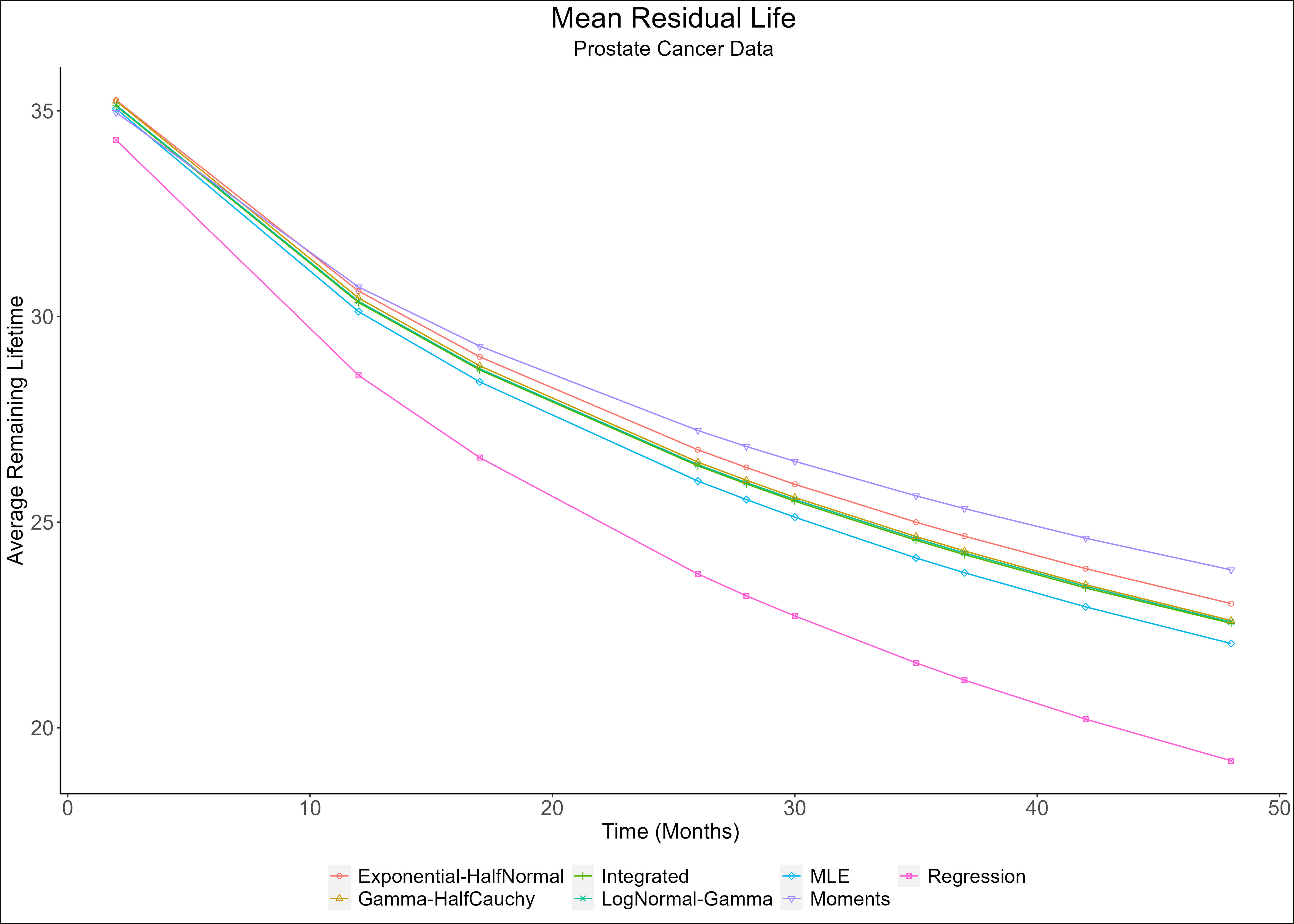}  }
\caption{Mean residual lifetime plot comparing the top 3 MCMC-based models, the classic Weibull model, and the Integrated model for a random sample of 10 prostate cancer survival times. The MCMC methods for calculating the mean residual life are more robust than traditional methods. Regression and MLE methods underestimate the population mean remaining life, while the Moments overestimate the parameter.\label{fig: MRLPlotProstateData} }
\end{figure}
Table \ref{tab: RemainingLifetime} quantifies the deviation of each method from the integrated MRL model, which pools estimates from all \(27\) fitted models. The LogNormal-Gamma and Gamma-HalfCauchy priors yield the most reliable MRL estimates. MCMC methods slightly overestimate MRL, while MLE underestimates it, with regression showing the most significant bias, particularly for older patients.
\begin{sidewaystable}[htbp]
\centering
\caption{\label{tab: RemainingLifetime} A comparison of the mean residual lifetime estimates of the top 3 MCMC-based models, the classical Weibull models, and the Integrated model for a random sample of 10 prostate cancer survival times.
The integrated model uses the actual parameter values $\beta = 1.43$ and $\alpha =40.34$, pooled from the results of all the models under the study.
MCMC and Classic results represent percent deviations from the true mean residual lifetime of the integrated model.}
\centering
\fontsize{8}{10}\selectfont
\begin{tabular}[t]{@{}cccccccc@{}}
\toprule
\multicolumn{2}{c}{ } & \multicolumn{3}{c}{MCMC} & \multicolumn{3}{c}{Classics} \\
\cmidrule(l{3pt}r{3pt}){3-5} \cmidrule(l{3pt}r{3pt}){6-8}
time & integrated & LogNormal-Gamma & Gamma-HalfCauchy & Exponential-HalfNormal & MLE & Regression & Moments\\
\hline
2 & 35.12 & +0.06\% & +0.34\% & +0.40\% & -0.20\% & -2.36\% & -0.46\%\\

12 & 30.34 &  +0.10\% & +0.36\% & +0.92\% & -0.73\% & -5.83\% & +1.25\%\\

17 & 28.70 & +0.10\% & +0.35\% & +1.11\% & -1.01\% & -7.42\% & +2.02\%\\

26 & 26.37 & +0.11\% & +0.34\% & +1.48\% & -1.40\% & -9.97\% & +3.26\%\\

28 & 25.93 & +0.12\% & +0.35\% & +1.54\% & -1.47\% & -10.49\% & +3.51\%\\

30 & 25.51 & +0.16\% & +0.35\% & +1.61\% & -1.53\% & -10.94\% & +3.80\%\\

35 & 24.56 & +0.16\% &  +0.37\% & +1.79\% & -1.75\% & -12.13\% & +4.40\%\\

37 & 24.21 & +0.17\% & +0.37\% & +1.86\% & -1.82\% & -12.60\% & +4.63\%\\

42 & 23.40 & +0.17\% & +0.34\% & +2.01\% & -1.97\% & -13.63\% & +5.17\%\\

48 & 22.54 & +0.13\% & +0.31\% & +2.13\% & -2.17\% & -14.82\% & +5.77\%\\
\bottomrule
\end{tabular}
\end{sidewaystable}

LogNormal-Gamma and Gamma-HalfCauchy priors offer more dependable estimates of Weibull parameters and MRL than traditional methods. In contrast to MLE, which relies on point estimates sensitive to sample variability, MCMC effectively utilizes prior information to mitigate bias and enhance efficiency. This advantage is particularly valuable when modeling survival data that exhibit complex hazard rate structures.

\newpage
\clearpage
\section{Conclusion and Discussion}\label{Conclussions_study_findings}

We introduced an adaptive MCMC modeling technique for lifetime data, utilizing a concise set of default prior distributions. This approach enables precise characterization of key survival time measures, such as the mean remaining lifetime, particularly for patients with terminal illnesses like prostate cancer. Our findings show that prior selection depends on sample size and hazard rate patterns. However, our algorithm adaptively identifies the most effective shape and scale prior combinations for any given dataset.

\begin{center}
\begin{threeparttable}[!h]
\caption{\label{tab: RecomendPriorsDHR} A list of recommended combinations of prior distributions for small and large Weibull-distributed data sets with decreasing hazard rate properties. The priors are ranked based on the AWRE.}
\fontsize{8}{10}\selectfont
\begin{tabular*}{\textwidth}{@{\extracolsep\fill}llll}
\toprule
\multicolumn{2}{c}{n = (15, 25)} & \multicolumn{2}{c}{n = (55, 100)} \\
\cmidrule(l{3pt}r{3pt}){1-2} \cmidrule(l{3pt}r{3pt}){3-4}
Shape, $\beta$ & Scale, $\alpha$ & Shape, $\beta$ & Scale, $\alpha$\\
\midrule
HalfNormal &  & Exponential & \\

LogNormal &  & HalfNormal & \\

Gamma & \multirow[t]{-3}{*}{\raggedright\arraybackslash  \footnotemark[1] HalfCauchy} & Gamma & \\

HalfNormal & LogNormal & LogNormal & \multirow[t]{-4}{*}{\raggedright\arraybackslash  \footnotemark[1] HalfCauchy}\\
\cmidrule{1-4}
Gamma &  & LogNormal & \\

LogNormal & \multirow[t]{-2}{*}{\raggedright\arraybackslash Gamma} & HalfNormal & \multirow[t]{-2}{*}{\raggedright\arraybackslash Gamma}\\
\bottomrule
\end{tabular*}
\begin{tablenotes}
\item [1] The HalfCauchy prior distribution for the scale parameter creates an effective model for large and small data sets.
\end{tablenotes}
\end{threeparttable}
\end{center}

Tables \ref{tab: RecomendPriorsDHR} and \ref{tab: RecomendPriorsIHR} summarize our recommended prior distributions, ranked by AWRE. While no prior combination is universally optimal, the Half-Cauchy distribution emerges as a robust choice for the scale parameter. Compared to classical methods, the MCMC approach consistently outperforms across diverse lifetime datasets, regardless of structure. It also mitigates overdispersion issues commonly encountered in classical techniques, leading to more accurate survival data analysis. Classical methods with small sample sizes are unreliable but improve as data size increases.

\begin{center}
\begin{threeparttable}[!h]
\caption{\label{tab: RecomendPriorsIHR} A list of recommended combinations of prior distributions for small and large Weibull-distributed data sets with increasing hazard rate properties. The priors are ranked based on the AWRE.}
\fontsize{8}{10}\selectfont
\begin{tabular*}{\textwidth}{@{\extracolsep\fill}llll}
\toprule
\multicolumn{2}{c}{n = (15, 25)} & \multicolumn{2}{c}{n = (55, 100)} \\
\cmidrule(l{3pt}r{3pt}){1-2} \cmidrule(l{3pt}r{3pt}){3-4}
Shape, $\beta$ & Scale, $\alpha$ & Shape, $\beta$ & Scale, $\alpha$\\
\midrule
Gamma &  & Exponential & \\

HalfNormal &  & LogNormal & \\

LogNormal & \multirow[t]{-3}{*}{\raggedright\arraybackslash \footnotemark[1] HalfCauchy} & Gamma & \\

 &  & HalfNormal & \multirow[t]{-4}{*}{\raggedright\arraybackslash \footnotemark[1] HalfCauchy}\\
\bottomrule
\end{tabular*}
\begin{tablenotes}
\item [1] The HalfCauchy prior distribution for the scale parameter creates an effective model for large and small data sets.
\end{tablenotes}
\end{threeparttable}
\end{center}

The Gamma, LogNormal, and HalfNormal distributions exhibit desirable properties for capturing latent information and producing accurate estimates for shape parameters in increasing and decreasing hazard rate models. The Gamma distribution is exceptionally versatile, while the LogNormal distribution effectively represents varying hazard rate patterns.

Our research has led to developing an adaptive MCMC-based model selection process that utilizes weighted relative efficiency to identify optimal default models. Recognizing that model parameters belong to positive sampling spaces, we compared this recommended MCMC model with classical approaches across various sample sizes and hazard rate structures. Our key finding is that MCMC models consistently achieve absolute efficiency, with sampling variance contributing the most to total asymptotic variance. This efficiency is critical for accurately fitting Weibull distribution models to datasets with varying hazard rate properties.

The practical advantages of MCMC models over classical methods underscore their relevance in statistical modeling, offering a reliable framework for survival data analysis. Our findings highlight the importance of adaptive prior selection in achieving robust and scalable Bayesian inference for lifetime data, making MCMC a powerful survival analysis tool for researchers and practitioners.

\section*{Declarations}

\subsection*{Competing interests}

The authors declare no competing interests in this article.

%

\begin{appendices}

\section{Algorithm \ref{alg: ALg1}}\label{Alg-1}

The proposed adaptive MCMC approach, detailed in Algorithm \ref{alg: ALg1}, iteratively updates the prior distribution to minimize the Kullback-Leibler (KL) divergence.

This methodology is beneficial when the initial prior is a poor fit for the data or when sampling from the posterior distribution is computationally challenging. By efficiently exploring the posterior space, the \textit{No-U-Turn} Sampler facilitates optimal prior selection, reducing divergence.

Mathematically, at each iteration, the adaptive MCMC updates the prior distribution to minimize KL divergence:

\[
\pi_{t+1}(\theta) = \arg\min_{\pi(\theta)} D_{\text{KL}}(P_{\text{true}} || P_{\text{posterior}}),
\]
where \( P_{\text{true}} \) is the true failure time distribution and \( P_{\text{posterior}} \) is the posterior predictive distribution given by:
\[
P_{\text{posterior}}(x | X_1, X_2, \dots, X_n) = \int f(x | \theta) \pi(\theta | X_1, X_2, \dots, X_n) d\theta.
\]

Since KL divergence is non-negative and decreases with each iteration of the adaptive MCMC procedure, the process converges to a minimum divergence value. At convergence, the prior distribution is optimized to approximate the true failure time density closely.

	\begin{algorithm}[H]\label{alg: ALg1}
		\caption{Adaptive MCMC with Prior Updating}
		\begin{algorithmic}[1]
			\State \textbf{Input:} Observed failure times \( X_1, X_2, \dots, X_n \), initial prior \( \pi_0(\theta) \)
			\State \textbf{Output:} Updated prior distribution \( \pi(\theta) \), posterior samples \( \theta_1, \theta_2, \dots, \theta_m \)
			\State
			\State \textbf{Step 1: Initial Prior Selection}
			\State \quad Start with an initial prior distribution \( \pi_0(\theta) \) over the model parameters \( \theta \) (e.g., the shape and scale parameters of the Weibull distribution for failure times).
			\State
			\State \textbf{Step 2: Sampling from the Posterior}
			\State \quad Using MCMC (such as NUTS), sample from the posterior distribution \( \pi(\theta | X_1, X_2, \dots, X_n) \).
			\State \quad Obtain \( \theta_1, \theta_2, \dots, \theta_m \), the posterior samples.
			\State
			\State \textbf{Step 3: Update Priors}
			\State \quad Update the prior distribution \( \pi(\theta) \) based on the posterior samples to reduce the KL divergence between the posterior predictive distribution and the true failure time density.
			\State
			\State \textbf{Repeat Steps 2 and 3 as necessary to refine the prior and posterior distributions.}
		\end{algorithmic}
	\end{algorithm}

\section{Proof Theorem \ref{theo2}}\label{Ap-theo2}


For a Weibull distribution with parameters \(\beta\) (shape) and \(\alpha\) (scale), the log-likelihood function is given by:
\[
\ell(\theta; \{X_i\}_{i=1}^n) = \sum_{i=1}^n \log f(X_i \mid \beta, \alpha) = n \log \beta - n \beta \log \alpha + (\beta - 1) \sum_{i=1}^n \log X_i - \sum_{i=1}^n \left( \frac{X_i}{\alpha} \right)^\beta.
\]

 The Fisher Information matrix is a \(2 \times 2\) matrix given by:
\[
I_n(\theta) = 
\begin{pmatrix}
	\mathbb{E}\left[-\frac{\partial^2 \ell(\theta)}{\partial \alpha^2}\right]  & \mathbb{E}\left[-\frac{\partial^2 \ell(\theta)}{\partial \beta \partial \alpha}\right] \\
	\mathbb{E}\left[-\frac{\partial^2 \ell(\theta)}{\partial \alpha \partial \beta}\right] & \mathbb{E}\left[-\frac{\partial^2 \ell(\theta)}{\partial \beta^2}\right]
\end{pmatrix}.
\]
The elements of the Fisher Information matrix can be computed as:
\[
I_n(\beta, \alpha) = 
\begin{pmatrix}
	\frac{n \beta^2}{\alpha^2}  & -\left[1 + \lambda_1 \right]\frac{n}{\alpha} \\
	 -\left[1 + \lambda_1 \right]\frac{n}{\alpha}  & \quad \left[1 + \lambda_2 + 2\lambda_1\right] \frac{n}{\beta^2}
\end{pmatrix},
\]
where \(\lambda_1 = -0.5772 \text{ (Euler's constant), } \lambda_2 = \lambda_1^2 + \frac{\pi^2}{6} = 1.9781\).
By the Central Limit Theorem for MLEs, under regularity conditions, the distribution of the MLE \(\hat{\theta}_n = (\hat{\beta}_n, \hat{\alpha}_n)\) converges in distribution to a multivariate normal distribution:
\[
\sqrt{n} (\hat{\theta}_n - \theta_0) \xrightarrow{d} N(0, I_n(\theta_0)^{-1}),
\]
where \(I_n(\theta_0)\) is the Fisher Information matrix evaluated at the true parameter values \(\theta_0 = (\beta_0, \alpha_0)\).

\section{Bayesian Approach}\label{BayesianMethod}

Let $\textbf{t} = \left\{t_1, \dots, t_n \right\}~\text{be a random sample with a sampling distribution }~ p\left(\bf{t}~|~\beta, \alpha, \eta,\lambda\right)$. The likelihood function is given by Equation \ref{eq: GeneralLike}, and the posterior distribution is expressed in Equation \ref{eq: BayesianGeneralModel}.

\begin{align}
		L\left( \beta, \alpha | ~\textbf{t} \right)  = \prod_{i=1}^{n}p\left(t_i~|~ \beta, \alpha, \eta, \lambda\right) \label{eq: GeneralLike}
\end{align}

\begin{align}
	p\left( \beta, \alpha, \eta, \lambda | ~\textbf{t} \right) & \propto \prod_{i=1}^{n}p\left(t_i~|~ \beta, \alpha, \eta, \lambda\right)p\left(\beta~|~\eta \right) p\left(\alpha~|~\lambda\right) p\left(\eta\right) p\left(\lambda\right) \label{eq: BayesianGeneralModel}
\end{align}
Bayesian models are sensitive to the choice of prior distributions in their conditional form $p(.|.)$ and marginal form $p(.)$. Ensuring a suitable selection of priors is crucial for stable and interpretable posterior estimates.

Equation \ref{eq: LimitNormalDist} represents the limiting form of a standard normal distribution as $\sigma \rightarrow \infty$, demonstrating its relationship to the LogNormal density in Equation \ref{eq: LimitsLogNormalPrior}.

\begin{align}
	\lim_{\sigma \to \infty}\phi_{norm} \left( \frac{\log t - \mu}{\sigma}\right) &	= \lim_{\sigma \to \infty} \frac{1}{\sqrt{2 \pi}}\exp\left\{ - \frac{1}{2}\left( \frac{\log t - \mu}{\sigma}\right)^2\right\} \nonumber\\
																								&	= \frac{1}{\sqrt{2\pi} } \in \left(0, \infty \right) \label{eq: LimitNormalDist}
\end{align}

\begin{eqnarray}
	\frac{1}{\sigma t}\phi_{\mathit{norm}}\left(\frac{\log{t} - \mu}{\sigma} \right) \approx \frac{1}{\sigma t\sqrt{2 \pi}} \label{eq: LimitsLogNormalPrior}
\end{eqnarray}

We adopt a LogNormal hyperprior distribution that is weakly informative, allowing for a broad range of potential parameter values. Equation \ref{eq: LimitsLogNormalPrior} shows that for a LogNormal prior $\pi(t)$, its density is proportional to $1/t$ as the log standard deviation increases. This property reduces weight on extreme values while maintaining flexibility in the prior, making it well-suited for Bayesian modeling.

\subsection{Noninformative and reference priors}

The selection of priors has been widely studied in Bayesian analysis \cite{gelman1995bayesian, robert1999monte, brooks2011handbook, berger2013statistical}. Early studies \cite{jaynes1968prior, bernardo1979reference, berger1992development, kass1996selection} advocated for noninformative priors, such as Jeffreys' priors, to facilitate objective Bayesian inference. However, for multi-parameter models, noninformative priors can result in improper distributions, leading to extreme posterior values. Studies \cite{sun1997note, sun1998reference} demonstrated that reference priors often provide better frequentist coverage properties than Jeffreys' priors for small sample cases.

\subsection{Weakly informative reference priors}

Recent research \cite{gelman2008weakly, gelman2017prior, lemoine2019moving, depaoli2020importance, tian2023specifying} has shifted towards weakly informative priors, which balance flexibility with regularization. Weakly informative priors place reasonable density on plausible parameter values while limiting extreme estimates. \cite{gelman1995bayesian} argues that weakly informative priors incorporate minimal real-world information to improve interpretability while avoiding the pitfalls of noninformative priors.


\begin{definition}\label{DefIncompleteGamma} Incomplete Gamma Functions
	
	In accordance with the notation presented in section 8.2 of \cite{olver2010nist}, the incomplete Gamma functions, denoted as $\Gamma{\left(a,z\right)}$ and $\Gamma{\left(a,z\right)}$, are defined as follows:
	\begin{align}
		\label{IncompleteGammaFunc}
		\Gamma{\left(a,z\right)} &= \int_{0}^{z}t^{a-1}{\exp}\{-t\}dt, \quad \text{and} \quad \Gamma{\left(a,z\right)} = \int_{z}^{\infty}t^{a-1}{\exp}\{-t\}dt, \quad a \in \mathfrak{R}, ~ z \ge 0
	\end{align}
	It is noteworthy that the relationship $\Gamma{\left(a,z\right)} + \Gamma{\left(a,z\right)} = \Gamma{\left(a\right)}$ holds for $a \ne 0, -1, -2, \dots$. This relationship can be alternatively expressed as:
	\begin{align}
		\label{NormalizedGammFunc}
		\Gamma{\left(a,z\right)} = \Gamma{\left(a\right)}\left(1 - p\left(a, x\right) \right)
	\end{align}
	where $p\left(a, x\right) = \nicefrac{\Gamma{\left(a,z\right)}}{\Gamma{\left(a\right)}}$, representing the cumulative Gamma distribution with a unit scale. The ordinary Gamma function, denoted as $\Gamma{\left(a\right)}$, is defined as $\Gamma{\left(a\right)} = \int_{0}^{\infty}t^{a-1}{\exp}\{-t\}dt$.
\end{definition}

This definition encapsulates the integral representations and relationships characterizing the incomplete Gamma functions, offering a foundational understanding for subsequent discussions in the context of the provided mathematical expressions.

\subsection{Quantities of Interest}
In survival analysis, mean residual lifetime (MRL) and mortality risk (MR) are essential for evaluating treatment effectiveness.
Let $F\left( x\right)$ be the $\mathit{CDF}$ of a continuous random variable $\mathbf{X}$. Assuming that the survival function $\overline{F}\left( x\right) = 1 - F\left( x\right)$ is well defined, then MR, hazard rate
$h\left(x\right) = \nicefrac{dF\left(x\right)}{F\left(x\right)}$.
Let the intervention time be $t=0$, and the system still functions at time $x$. Suppose the system fails at a random time $T$. Then, the probability that the system survives an additional time $t$ is given by $\nicefrac{\overline{F}\left(x + t\right)}{\overline{F}\left( x\right)}, ~ 0< x< t$.
Mean residual lifetime can then be defined as
\begin{eqnarray}\label{MeanResidLife}
MRL\left(x\right) = \nicefrac{\int_{0}^{\infty} \overline{F}\left(t + x\right) dt}{\overline{F}\left(x\right)}
\end{eqnarray}

Let $F\left(x\right)$ be a cumulative distribution function of the Weibull distribution with parameters shape $(\beta)$ and scale $(\alpha)$. Then, we can show that
\begin{align}
\label{eq: MeanResidualLife}
MRL\left( x\right) = & \frac{\int_{0}^{\infty} \exp\left\{-\left(\nicefrac{x + t}{\alpha}\right)^\beta \right\}dt}{\exp\left\{-\left(\nicefrac{x}{\alpha}\right)^\beta\right \}} \nonumber\\
= & \frac{\alpha}{\beta}\exp\left\{ \left(\nicefrac{x}{\alpha}\right)^\beta\right\}\Gamma \left( \nicefrac{1}{\beta}, \left(\nicefrac{x}{\alpha} \right)^\beta \right) \nonumber \\
\textit{by Definition \ref{DefIncompleteGamma}, Equation \ref{NormalizedGammFunc},} \nonumber \\
= & \frac{\alpha}{\beta}\exp\left\{ \left(\nicefrac{x}{\alpha}\right)^\beta\right\}\Gamma\left(\nicefrac{1}{\beta}\right)\left\{ 1 - \tilde{F}\left( \left(\nicefrac{x}{\alpha} \right)^\beta; \nicefrac{1}{\beta}, 1 \right)\right\}
\end{align}
where,
\begin{eqnarray}\label{CumGammFunc}
\tilde{F}\left( \left(\nicefrac{x}{\alpha} \right)^\beta; \nicefrac{1}{\beta}, 1 \right) = \int_{0}^{\left(\nicefrac{x}{\alpha} \right)^\beta}u^{\frac{1}{\beta}-1}\exp\{ -u\}du
\end{eqnarray}

For given values of $\beta$ and $\alpha$, Equation \ref{CumGammFunc} represents the cumulative probability function of a Gamma distribution, allowing for analytical solutions of Equation \ref{eq: MeanResidualLife}.


\section{MCMC Method} \label{MCMCMethods}

MCMC methods are fundamental to Bayesian inference, particularly for sampling from complex posterior distributions. In this section, we explore non-conjugate priors for the Weibull distribution parameters (Table \ref{tab: DefinitionOfPriors}) to determine the optimal parameter values, $\theta_{\text{opt}}$, for the sampling distribution. These values are chosen from the set of all possible values, denoted as $\Omega_{\theta}$. By leveraging MCMC algorithms \cite{gelman1995bayesian, robert1999monte, brooks2011handbook}, we can approximate probability distributions even when their densities are too complex for explicit integration.

Let $\theta = \left(\beta, \alpha\right)$ be parameters of the Weibull distribution. Let Equation (\ref{eq: BayesWeibLikelihood}) be the likelihood function of $\theta$ given the observed lifetime $\mathbf{t} = \{t_1, t_2, \dots, t_n\}$. 

\begin{eqnarray}
\left( ~ t_1, t_2, \dots, t_n ~| ~ \theta ~ \right) = \prod_{i=1}^{n} \left( ~ t_i ~ | ~ \theta_i ~ \right) \label{eq: BayesWeibLikelihood}
\end{eqnarray}

Assuming $\gamma$ is the hyperparameter of the prior distribution, $\left( \theta \right)$. The joint posterior distribution of the parameters is given by:  $\left( ~\theta, \gamma ~| ~\mathbf{t} ~ \right) \propto \left(~ \mathbf{t} ~|~ \theta ~\right) \left(~ \theta ~| ~\gamma ~\right) \left(~ \gamma~\right)$. 
To make Bayesian inferences about $\theta$, we use its marginal posterior distribution:

\begin{eqnarray}
\left(~ \theta~ | ~ t_1, t_2, \dots, t_n ~\right) = \int \left(~ \theta, \gamma ~ |~ t_1, t_2, \dots, t_n~ \right) d \gamma \label{eq: marginalposterior}
\end{eqnarray}

Since obtaining a closed-form solution for the posterior mean, $\mathit{E}[ \theta | \mathbf{t} ]$, is often infeasible, MCMC algorithms \cite{hoffman2014no, brooks2011handbook, robert1999monte, carpenter2017stan} provide an effective alternative by enabling direct sampling from the joint posterior distribution $\left( \theta, \gamma~ |~ \mathbf{t} ~\right)$ and thus from the marginal posterior distributions $\left( \theta~ | ~ \mathbf{t}\right)$ and $\left(~\gamma~ | ~ \mathbf{t}~\right)$ as those of Weibull model.. This allows us to estimate summary statistics using Monte Carlo approximations, leveraging Markov chain properties under ergodicity and independence conditions.

\subsubsection{MCMC Sampling Statements}

We specify prior distributions using the moment-matching method, ensuring flexibility by introducing hyperpriors. Weakly informative priors are set to allow sufficient prior diffusion while maintaining computational stability. For instance, LogNormal hyperpriors incorporate broad shape parameters, while Normal hyperpriors include more extensive variance terms (Tables \ref{tab: ShapeScalePriorSpecs} and \ref{tab: ScaleOnlyPriorSpecs}).

We adopt a complete probabilistic approach to parameter estimation rather than traditional point estimates. This approach quantifies uncertainty across the entire parameter space, yielding more robust inference.


\begin{center}
\begin{threeparttable}
\caption{\label{tab: ShapeScalePriorSpecs}Specifications of prior distributions for the Weibull model parameters. LogNormal hyper-priors provide a convenient structure for hierarchical modeling.}
\centering
\fontsize{8}{10}\selectfont
\begin{tabular}{@{}llll@{}}
\toprule
Exponential & Gamma & HalfNormal & LogNormal\\
\midrule
$\footnotemark[1] \lambda_{\beta} \sim \mathit{LogNormal}\left(0, 25\right)$       &
$\footnotemark[1] \mathit{A}_{\beta} \sim \mathit{LogNormal}\left(0,25\right)$    &
$\footnotemark[1]\mu_{\beta} \sim \mathit{Normal}\left(0,100\right)$           &
$\footnotemark[1]\mu_{\beta} \sim \mathit{Normal}\left(0,100\right)$\\
$\footnotemark[1] \lambda_{\alpha} \sim \mathit{LogNormal}\left(0,25\right)$      &
$\footnotemark[1] \mathit{B}_{\beta} \sim \mathit{LogNormal}\left(0,25\right)$    &
$\footnotemark[1]\sigma_{\beta} \sim \mathit{LogNormal}\left(0,25\right)$        &
$\footnotemark[1]\sigma_{\beta} \sim \mathit{LogNormal}\left(0,25\right)$\\
$\quad \beta \sim \mathit{Exponential}\left(\footnotemark[1] \lambda_{\beta}\right)$     &
$\footnotemark[1] \mathit{A}_{\alpha} \sim \mathit{LogNormal}\left(0,25\right)$           &
$\footnotemark[1]\mu_{\alpha} \sim \mathit{Normal}\left(0,100\right)$                  &
$\footnotemark[1]\mu_{\alpha} \sim \mathit{Normal}\left(0,100\right)$\\
$\quad \alpha \sim \mathit{Exponential}\left(\footnotemark[1] \lambda_{\alpha}\right)$   &
$\footnotemark[1] \mathit{B}_{\alpha} \sim \mathit{LogNormal}\left(0,25\right)$           &
$\footnotemark[1]\sigma_{\alpha} \sim \mathit{LogNormal}\left(0,25\right)$               &
$\footnotemark[1]\sigma_{\alpha} \sim \mathit{LogNormal}\left(0,25\right)$\\
&
$\quad \beta \sim \mathit{Gamma}\left(\footnotemark[1]\mathit{A}_{\beta}, \footnotemark[1]\mathit{B}_{\beta}\right)$ &
$\quad \beta \sim \mathit{HalfNormal}\left(\footnotemark[1]\mu_{\beta}, \footnotemark[1]\sigma_{\beta}\right)$       &
$\quad \beta \sim \mathit{LogNormal}\left(\footnotemark[1] \mu_{\beta}, \footnotemark[1] \sigma_{\beta}\right)$\\
&
$\quad\alpha \sim \mathit{Gamma}\left(\footnotemark[1]\mathit{A}_{\alpha}, \footnotemark[1]\mathit{B}_{\alpha}\right)$
&
$\quad\alpha \sim \mathit{HalfNormal}\left(\footnotemark[1]\mu_{\alpha}, \footnotemark[1]\sigma_{\alpha}\right)$
&
$\quad \alpha \sim \mathit{LogNormal}\left(\footnotemark[1] \mu_{\alpha}, \footnotemark[1] \sigma_{\alpha}\right)$\\
\bottomrule
\end{tabular} 
\begin{tablenotes}
\item [1] Weakly informative hyperparameter priors are initialized based on typical values within physically possible units of measurement, but on a wide range of values. By setting large values for the shape parameter, log-standard deviation of the LogNormal distribution, and significant variance for the Normal distribution, the hierarchical structure allows minimal information through the hyperparameters for posterior regularization.
\end{tablenotes}
\end{threeparttable}
\end{center}

\bigskip

\begin{center}
\begin{threeparttable}[h!]
\caption{Specifications of prior distributions used only for the scale parameter of the Weibull life model}
\label{tab: ScaleOnlyPriorSpecs}
\fontsize{8}{10}\selectfont
\begin{tabular}[t]{ll}
\toprule
HalfCauchy & InverseGamma\\
\midrule
$\textsuperscript{1}\mu_{\alpha}  \sim \mathit{Normal}\left(0,100\right)$     &
$\textsuperscript{1}\mathit{A}_{\alpha} \sim  \mathit{LogNormal}\left(0,25\right)$\\
$\textsuperscript{1} \sigma_{\alpha} \sim \mathit{LogNormal}\left(0,25\right)$  &
$\textsuperscript{1} \mathit{B}_{\alpha} \sim \mathit{LogNormal}\left(0,25\right)$\\
$\quad \alpha \sim \mathit{HalfCauchy}\left( \textsuperscript{1}\mu_{\alpha}, \textsuperscript{1} \sigma_{\alpha}\right)$   &
$\quad \alpha \sim \mathit{InverseGamma}\left(\textsuperscript{1}\mathit{A}_{\alpha}, \textsuperscript{1}\mathit{B}_{\alpha}\right)$\\
\bottomrule
\end{tabular}
\begin{tablenotes}
\item [1] Weakly informative hyperparameter priors are initialized based on typical values within physically possible units of measurement, but on a wide range of values. By setting large values for the shape parameter, log-standard deviation of the LogNormal distribution, and significant variance for the Normal distribution, the hierarchical structure allows minimal information through the hyperparameters for posterior regularization.
\end{tablenotes}
\end{threeparttable}
\end{center}

\begin{table}[h!]
	\caption{\label{tab: DefinitionOfPriors} 
		The proposed initial distribution is either assigned to $\beta$ (shape), $\alpha$ (scale), or independently for $\theta = (\beta, \alpha)$ both shape and scale.}
	\centering
	\begin{tabular}{@{}lll@{}}
		\toprule
		\shortstack{Priors\\ Distributions} & p.d.f of Priors & \shortstack{Initial Values of\\ Priors Parameters}\\
		\midrule
		$\mathit{Gamma}$        & $\pi_{\theta} = \left(\nicefrac{b^a}{\Gamma a} \right) \theta ^{a- 1} \exp -(b \theta)$   & $a = \nicefrac{E^2\left[ \theta\right]}{\sigma^2_{\theta}}$ \\
		&&$b = \nicefrac{E\left[ \theta \right]}{\sigma^2_{\theta}}$  \\ \\
		$\mathit{Exponential}$  & $\pi_{\theta} = b ~ \exp (-b \theta)$                                 & $b = \nicefrac{1}{E\left[\theta \right]}$\\ \\
		$\mathit{LogNormal}$    & $C = \nicefrac{1}{\sigma \sqrt{2 \pi}}$                                          &\\
		& $\pi_{\theta} = C \exp ( \nicefrac{ -(\ln \theta - \mu)^2}{ 2 \sigma^2} )$      & $\mu = \ln\left( \nicefrac{E\left[\theta \right]}{ \sqrt{\nicefrac{\sigma^2_{\theta}}{E^2\left[ \theta\right]} + 1 } }\right)$\\ \\
		&&$\sigma = \sqrt{\ln\left(\nicefrac{\sigma^2_{\theta}}{E^2\left[ \theta\right]} + 1\right)}$\\ \\
		$\mathit{HalfNormal}$     &$\pi_{\theta} = \nicefrac{1}{\sigma \sqrt{2 \pi}} ~ \exp \left(\nicefrac{-(\theta - \mu)^2}{2 \sigma ^2} \right)$    & $\theta > 0, ~ \mu = E\left[\theta \right], ~ \sigma = \sigma_{\theta}$ \\ \\
		$\mathit{HalfCauchy}$   &$C = \nicefrac{1}{\pi \sigma}$                                         &\\
		&$\pi_{\alpha} = C~ \left({1 + \nicefrac{\left( \alpha - \mu\right)^2}{\sigma^2}}\right)^{-1}$    &$\alpha > 0, ~ \mu = E\left[\alpha \right],~ \sigma = \sigma_{\alpha}$\\ \\
		$\mathit{InverseGamma}$ &$\pi_{\alpha} = \left(\nicefrac{b^a}{\Gamma a} \right) \alpha ^ {- a- 1} \exp \left( \nicefrac{-b}{\alpha} \right)$& $a = \nicefrac{E^2\left[ \alpha\right]}{\sigma^2_{\alpha}} + 2$\\
		&&$b = E\left[\alpha \right] + \nicefrac{E^3\left[\alpha \right]}{\sigma^2_{\alpha}}$\\
		\midrule
	\end{tabular}	
\end{table}

\end{appendices}

\newpage
\clearpage



\bibliography{bibliography}

\end{document}